\documentclass{llncs}

\usepackage{array}
\usepackage[dvips]{graphicx}
\usepackage{epic, eepic}
\usepackage{amssymb,amsmath,stmaryrd}
\usepackage{latexsym}
\usepackage{subfigure}
\usepackage[usenames]{color}
\usepackage{gastex}
\usepackage{comment}
\usepackage{paralist}
\usepackage{multirow}
\usepackage{slashbox}

\usepackage{times}

\let\emptyset\varnothing

\newcounter{compressEnum}

{}%\end{paragraph}}

\newtheorem{conj}{Conjecture} 

\newtheorem{theo}{Theorem} 

\newtheorem{lemm}{Lemma}

\newtheorem{coro}{Corollary}

\newtheorem{defi}{Definition}

\makeatletter
\let\c@lemm\c@theo
\let\c@coro\c@theo
\let\c@defi\c@theo
\let\c@assu\c@theo
\makeatother

\def\abs#1{\ensuremath{\lvert #1\rvert}}

 \newcommand{\B}{{\mathcal{B}}}

\newcommand{\Reach}{\mathsf{Reach}}
\newcommand{\target}{{\cal T}}

\newcommand{\Safe}{\mathsf{Safe}}

\newcommand{\Parity}{\mathsf{Parity}}

\newcommand{\real}{{\mathbb R}}

\newcommand{\nat}{\mathbb N} 
\newcommand{\tuple}[1]{\langle #1 \rangle}

\newcommand{\Obs}{\mathcal{O}}
\newcommand{\obs}{\mathsf{obs}}

\newcommand{\head}{\mathsf{head}}
\newcommand{\z}{\mathsf{z}}

\newcommand{\Inf}{\mathsf{Inf}}

\newcommand{\A}{\mathcal{A}}

\newcommand{\true}{{\sf true}}
\newcommand{\false}{{\sf false}}

\newcommand{\Prb}{\mathrm{Pr}}

\newcommand{\post}{\mathsf{post}}

\newcommand{\straa}{{\sigma}}
\newcommand{\Straa}{{\Sigma}}
\newcommand{\Alphabet}{{\Sigma}}
\newcommand{\dist}{{\cal D}}

\renewcommand{\partial}{{\sf partial}}

\newcommand{\Uniform}{{\sf Uniform}}

\title{{\bf Games with a Weak Adversary}\thanks{This research was partly
supported by Austrian Science Fund (FWF) Grant No P23499- N23, FWF NFN Grant No S11407-N23 (RiSE), ERC Start grant (279307: Graph Games),
Microsoft Faculty Fellowship Award, and European project Cassting (FP7-601148).}%$^,$\thanks{Fuller version:~\cite{}.}
}

\author{Krishnendu Chatterjee \inst{1}
\and Laurent Doyen  \inst{2} }

\institute{IST Austria
\and LSV, ENS Cachan \& CNRS, France 
}

\begin{document}
\maketitle
\pagestyle{plain}

\begin{abstract}
We consider multi-player graph games with partial-observation and parity objective.
While the decision problem for three-player games with a coalition of 
the first and second players against the third player is undecidable in general,    %%% L: added "in general"
we present 
a decidability result for partial-observation games where the first and third player
are in a coalition against the second player, thus where the second player
is adversarial but weaker due to partial-observation. 
We establish tight complexity bounds in the case where player~$1$ is 
less informed than player~$2$, namely 2-EXPTIME-completeness for parity objectives.
The symmetric case of player~$1$ more informed than player~$2$ is much more
complicated, and we show that already in the case where player~$1$ has 
perfect observation, memory of size non-elementary is necessary in general for 
reachability objectives, and the problem is decidable for safety and reachability objectives. 
Our results
have tight connections with partial-observation stochastic games for which
we derive new complexity results.
\end{abstract}

\section{Introduction}

\noindent{\bf Games on graphs.}
Games played on graphs are central in several important problems in computer 
science, such as reactive synthesis~\cite{PR89,RamadgeWonham87}, 
verification of open systems~\cite{AHK02}, and many others.
The game is played by several players on a finite-state graph, 
with a set of angelic (existential) players and a set of demonic (universal) 
players as follows:
the game starts at an initial state, and given the current state, the 
successor state is determined by the choice of moves of the players.
The outcome of the game is a \emph{play}, which is an infinite sequence of 
states in the graph. 
A \emph{strategy} is a transducer to resolve choices in a game for a player 
that given a finite prefix of the play specifies the next move.
Given an objective (the desired set of behaviors or plays), the goal of the 
existential players is to ensure the play belongs to the objective 
irrespective of the strategies of the universal players.
In verification and control of reactive systems an objective is 
typically an $\omega$-regular set of paths. 
The class of $\omega$-regular languages, that extends classical regular languages to 
infinite strings, provides a robust specification language to express
all commonly used specifications, and parity objectives are a canonical way to 
define such $\omega$-regular specifications~\cite{Thomas97}.
%%In a parity objective, every state of the game is mapped to a non-negative 
%%integer priority and the goal is to ensure that the minimum priority visited 
%%infinitely often is even.
Thus games on graphs with parity objectives provide a general framework for 
analysis of reactive systems.

\smallskip\noindent{\bf Perfect vs partial observation.}
Many results about games on graphs make the hypothesis of {\em perfect 
observation} (i.e., players have perfect or complete observation about the 
state of the game).  
In this setting, due to determinacy (or switching of the strategy quantifiers 
for existential and universal players)~\cite{Martin75}, 
the questions expressed by an arbitrary alternation of quantifiers 
reduce to a single alternation, and thus are equivalent to solving
two-player games (all the existential players against all the universal players).
However, the assumption of perfect observation is often not realistic in practice.
For example in the control of physical systems, digital sensors with finite precision
provide partial information to the controller about the system state~\cite{DDR06,HK99}.
Similarly, in a concurrent system the modules expose partial interfaces
and have access to the public variables of the other processes, but not 
to their private variables~\cite{Reif84,AHK02}.
Such situations are better modeled in the more general framework
of \emph{partial-observation} games~\cite{Reif79,Reif84,RP80}.
%%and have been studied in the context of verification and synthesis.

\smallskip\noindent{\bf Partial-observation games.}
Since partial-observation games are not determined, unlike the perfect-observation 
setting, the multi-player games problems do not reduce to the case of two-player games. 
Typically, multi-player partial-observation games are studied in the following 
setting: a set of partial-observation existential players, against a perfect-observation 
universal player, such as for distributed synthesis~\cite{PR89,FS10,RS10}.
The problem of deciding if the existential players can ensure a reachability 
(or a safety) objective is undecidable in general, even for two existential 
players~\cite{RP79,PR89}.
However, if the information of the existential players form a chain 
(i.e., existential player~1 more informed than existential player~2, 
existential player~2 more informed than existential player~3, and so on),
then the problem is decidable~\cite{PR89,MT01,MW03}.

\smallskip\noindent{\bf Games with a weak adversary.} 
One aspect of multi-player games that has been largely ignored is the 
presence of weaker universal players that do not have perfect observation.
However, it is natural in the analysis of composite reactive systems that some 
universal players represent components that do not have access to all variables
of the system.
In this work we consider games where adversarial players can have partial 
observation.
If there are two existential (resp., two universal) players with incomparable 
partial observation, then the undecidability results follows from~\cite{RP79,PR89};
and if the information of the existential (resp., universal) players form 
a chain, then they can be reduced to one partial-observation existential 
(resp., universal) player.
We consider the following case of partial-observation games: one 
partial-observation existential player (player~1), 
one partial-observation universal player (player~2), 
one perfect-observation existential player (player~3), and 
one perfect-observation universal player (player~4).   
Roughly, having more partial-observation players in general leads to undecidability, 
and having more perfect-observation players reduces to two perfect-observation players.
We first present our results and then discuss two applications of our 
model.

\smallskip\noindent{\bf Results.} Our main results are as follows:
\begin{compactenum}
\item \emph{Player~1 less informed.} We first consider the case when player~1
is less informed than player~2. 
We establish the following results: $(i)$~a 2-EXPTIME upper bound 
for parity objectives and a 2-EXPTIME lower bound for reachability objectives
(i.e., we establish 2-EXPTIME-completeness); 
$(ii)$~an EXPSPACE upper bound for parity objectives when player~1 is blind 
(has only one observation), and EXPSPACE lower bound for reachability 
objectives even when both player~1 and player~2 are blind.
In all these cases, if the objective can be ensured then the upper bound on 
memory requirement of winning strategies is at most doubly exponential.

\item \emph{Player~1 is more informed.} We consider the case when player~1
can be more informed as compared to player~2, and show that even when 
player~1 has perfect observation there is a non-elementary lower bound on the 
memory required by winning strategies.  
This result is also in sharp contrast to distributed games, where if only one 
player has partial observation then the upper bound on memory of winning 
strategies is %%at most 
exponential.
\end{compactenum}

\smallskip\noindent{\bf Applications.}
We discuss two applications of our results: the sequential 
synthesis problem, and new complexity results for partial-observation 
\emph{stochastic} games.

\begin{compactenum}
\item The sequential synthesis problem consists of a set of partially implemented modules,
where first a set of modules needs to be refined, followed by a refinement of some modules
by an external source, and then the remaining modules are refined so that the 
composite open reactive system satisfies a specification. 
Given the first two refinements cannot access all private variables, we have a 
four-player game where the first refinement corresponds to player~1, the 
second refinement to player~2, the third refinement to player~3, and 
player~4 is the environment.

\item In partial-observation stochastic games, there are two 
partial-observation players (one existential and one universal) playing in 
the presence of uncertainty in the transition function 
(i.e., stochastic transition function). 
The qualitative analysis question is to decide the 
existence of a strategy for the existential player to ensure the parity objective 
with probability~1 (or with positive probability) against all strategies of the 
universal player. 
The witness strategy can be randomized or deterministic (pure).
While the qualitative problem is undecidable, the practically relevant restriction to 
finite-memory pure strategies reduces to the four-player game problem.
Moreover, for finite-memory strategies, the decision problem for randomized strategies 
reduces to the pure strategy question~\cite{CD12}. 
By the results we establish in this paper, new decidability and complexity results are 
obtained for the qualitative analysis of partial-observation stochastic games
with player~$2$ partially informed but more informed than player~$1$.
The complexity results for almost-sure winning are summarized in Table~\ref{tab:complexity}.
Surprisingly for reachability objectives, whether 
player~2 is perfectly informed or more informed than player~1 does not change 
the complexity for randomized strategies, but it results in an exponential 
increase in the complexity for pure strategies.
%{\bf KRISH: MAYBE SOME MORE EXPLANATION.}
%\mynote{L: to be updated. KRISH: Updated.}

\end{compactenum}

\begin{table*}[!t]
\begin{center}
\scalebox{0.75}{
%\begin{scriptsize}
\begin{tabular}{|l|c|c|c|c|c|c|}
\cline{2-7}
\multicolumn{1}{l}{}              & \multicolumn{2}{|c|}{Reachability}          & \multicolumn{2}{|c|}{Parity}          & \multicolumn{2}{|c|}{Parity} \\
\cline{1-1}
\multirow{2}{*}{\backslashbox{Player~$1$}{Player~$2$}}               & \multicolumn{2}{|c|}{Finite- or infinite-memory strategies} & \multicolumn{2}{|c|}{Infinite-memory strategies} & \multicolumn{2}{|c|}{Finite-memory strategies}   \\
\cline{2-7}
\multicolumn{1}{|l|}{{\small \strut}} & Perfect     & More informed                  &  Perfect     & More informed              & Perfect     & More informed     \\
\hline
Randomized {\small \strut}           & EXP-c~\cite{CDHR07}   & EXP-c~\cite{BGG09}   &  Undec.~\cite{BBG08,CDGH10} &  Undec.~\cite{BBG08,CDGH10}  & EXP-c~\cite{CDNV14}   & {\bf 2EXP}     \\
\hline
Pure {\small \strut}           & EXP-c~\cite{CD12}   & {\bf 2EXP-c}   &  Undec.~\cite{BBG08} &  Undec.~\cite{BBG08} & EXP-c~\cite{CDNV14}   & {\bf 2EXP-c}     \\

\hline
\end{tabular}
%\end{scriptsize}
}
\end{center}

\caption{Complexity of qualitative analysis (almost-sure winning) for partial-observation stochastic games with partial observation for player~1 with reachability and parity objectives.
Player~2 has either perfect observation or more information than player~1(new results boldfaced).
For positive winning, all entries other than the first (randomized strategies
for player~~1 and perfect observation for player~2) remain the same, and the
complexity for the first entry for positive winning is PTIME-complete.
}\label{tab:complexity}

\end{table*}

\noindent{\bf Organization of the paper.}
In Section~\ref{sec:definitions} we present the definitions of 
three-player games, and other related models (such as partial-observation 
stochastic games). 
In Section~\ref{sec:player-one-less} we establish the results for 
three-player games with player~1 less informed, 
and in Section~\ref{sec:player-one-perfect} we show hardness of 
three-player games with perfect observation for player~1 (which is a special
case of player~1 more informed).
Finally, in Section~\ref{sec:more-than-three} we show how our upper bounds
for three-player games from Section~\ref{sec:player-one-less} extend to 
four-player games, and we discuss multi-player games.
We conclude with the applications in Section~\ref{sec:applications}.

\section{Definitions}\label{sec:definitions}

We first consider three-player games with parity objectives and we establish
new complexity results in Section~\ref{sec:player-one-less} that we later 
extend to four-player games in Section~\ref{sec:more-than-three}.
In this section, we also present the related models of alternating tree 
automata that provide useful technical results, and two-player
stochastic games for which our contribution implies new complexity results. 

\subsection{Three-player games}\label{sec:three-players}

%Define $\last(\cdot)$.

\paragraph{Games}
Given alphabets $A_i$ of actions for player~$i$ ($i=1,2,3$), a \emph{three-player game} 
is a tuple $G = \tuple{Q, q_0, \delta}$ where:
\begin{compactitem}
\item $Q$ is a finite set of states with
%\item $A_i$ is the finite alphabet of player~$i$ ($i=1,2,3$);
%%\item 
$q_0 \in Q$ the initial state; and
\item $\delta: Q \times A_1 \times A_2 \times A_3 \to Q$ is a deterministic 
transition function that, given a current state $q$, and actions $a_1 \in A_1$, $a_2 \in A_2$, $a_3 \in A_3$
of the players, gives the successor state $q' = \delta(q,a_1,a_2,a_3)$.
%\item
\end{compactitem}
The games we consider are sometimes called \emph{concurrent} because all three
players need to choose simultaneously an action to determine a successor 
state. The special class of \emph{turn-based} games corresponds to the case 
where in every state, one player has the turn and his sole action determines
the successor state. In our framework, a turn-based state for player~$1$
is a state $q \in Q$ such that $\delta(q,a_1,a_2,a_3) = \delta(q,a_1,a'_2,a'_3)$
for all $a_1 \in A_1$, $a_2,a'_2 \in A_2$, and $a_3,a'_3 \in A_3$. We define 
analogously turn-based states for player~$2$ and player~$3$.
A game is turn-based if every state of $G$ is turn-based (for some player).
The class of two-player games is obtained when $A_3$ is a singleton.
In a game~$G$, given $s \subseteq Q$, $a_1 \in A_1$, $a_2 \in A_2$, let
$\post^G(s,a_1,a_2,-) = \{q' \in Q \mid \exists q \in s \cdot \exists a_3 \in A_3: q' = \delta(q,a_1,a_2,a_3)\}$.

\paragraph{Observations}
For $i=1,2,3$, a set $\Obs_i \subseteq 2^Q$ of \emph{observations} (for player~$i$) 
is a partition of $Q$ (i.e., $\Obs_i$ is a set of non-empty and non-overlapping 
subsets of $Q$, and their union covers $Q$). Let $\obs_i: Q \to \Obs_i$ be the
function that assigns to each state $q \in Q$ the (unique) observation for player~$i$
that contains $q$, i.e. such that $q \in \obs_i(q)$.
The functions $\obs_i$ are extended to sequences $\rho = q_0 \dots q_n$ of states 
in the natural way, namely $\obs_i(\rho) = \obs_i(q_0) \dots \obs_i(q_n)$.
We say that player~$i$ is \emph{blind} if $\Obs_i = \{Q\}$, that is player~$i$ has only 
one observation; player~$i$ has \emph{perfect information} if $\Obs_i = \{\{q\} \mid q \in Q\}$, 
that is player~$i$ can distinguish each state; and
player~$1$ is \emph{less informed} than player~$2$ (we also say player~2 is more informed)
if for all $o_2 \in \Obs_2$, there exists $o_1 \in \Obs_1$ such that $o_2 \subseteq o_1$.

\paragraph{Strategies}
For $i=1,2,3$, let $\Straa_i$ be the set of \emph{strategies} $\straa_i: \Obs_i^+ \to A_i$ 
of player~$i$ that, given a sequence of past observations, give an action for player~$i$. 
Equivalently, we sometimes view a strategy of player~$i$ as a function 
$\straa_i: Q^+ \to A_i$ satisfying $\straa_i(\rho) = \straa_i(\rho')$ for all 
$\rho,\rho' \in Q^+$ such that $\obs_i(\rho) = \obs_i(\rho')$,
and say that $\straa_i$ is \emph{observation-based}.

%A strategy of player~$j$ ($j=2,3$) is a function $\straa_j: Q^+ \to A_j$
%without any restriction. We denote by $\Straa_j$ the set of strategies of 
%player~$j$.

\paragraph{Outcome}
Given strategies $\straa_i \in \Straa_i$ ($i=1,2,3$) in $G$, 
the \emph{outcome play} from a state $q_0$ is the infinite sequence $\rho^{\straa_1,\straa_2,\straa_3}_{q_0} = q_0 q_1 \dots$
such that for all $j \geq 0$, we have 
$q_{j+1} = \delta(q_j,a^j_1,a^j_2,a^j_3)$ where $a^j_i = \straa_i(q_0 \dots q_j)$ (for $i=1,2,3$).

\paragraph{Objectives}
An \emph{objective} is a set $\alpha \subseteq Q^\omega$ of infinite sequences
of states. 
A play $\rho$ \emph{satisfies} the objective $\alpha$ if $\rho \in \alpha$. 
An objective $\alpha$ is \emph{visible} for player $i$ if
for all $\rho, \rho' \in Q^\omega$, if $\rho \in \alpha$ and 
$\obs_i(\rho) = \obs_i(\rho')$, then $\rho' \in \alpha$. 
We consider the following objectives:
\begin{compactitem}

\item \emph{Reachability}.
Given a set $\target \subseteq Q$ of target states, the \emph{reachability} objective 
$\Reach(\target)$ requires that a state in $\target$ be visited at least once, 
that is, $\Reach(\target)= \{\rho = q_0 q_1 \dots  \mid \exists k \geq 0: q_k \in \target \}$. 
% Dually, the \emph{safety} objective $\Safe(\target)$ requires that only states 
% in $\target$ be visited.
% Formally,
% $\Safe(\target)=\{ \rho = q_0 q_1 \dots \mid \forall k \geq 0 \cdot q_k \in \target\}$.

\item \emph{Safety}.
Given a set $\target \subseteq Q$ of target states, the \emph{safety} objective 
$\Safe(\target)$ requires that only states in $\target$ be visited, 
that is, $\Safe(\target)= \{\rho = q_0 q_1 \dots  \mid \forall k \geq 0: q_k \in \target \}$. 

\item \emph{Parity}.
For a play $\rho = q_0 q_1 \dots$ we denote by $\Inf(\rho)$ the set of states 
that occur infinitely often in $\rho$, that is, $\Inf(\rho) = \{q \in Q \mid \forall k \geq 0 
\cdot \exists n \geq k: q_n = q\}$.
For $d \in \nat$, let $p: Q \to \{0,1,\dots,d\}$ be a priority
function, which maps each state to a nonnegative integer priority. The parity objective
$\Parity(p)$ requires that the minimum priority occurring infinitely often be even. Formally,
$\Parity(p) = \{\rho \mid \min \{p(q) \mid q \in \Inf(\rho)\} \text{ is even} \}$. 
Parity objectives are a canonical way to express $\omega$-regular objectives~\cite{Thomas97}.
If the priority function is constant over observations of player~$i$, 
that is for all observations $\gamma \in \Obs_i$ we have $p(q) = p(q')$ for all $q, q' \in \gamma$,
then the parity objective $\Parity(p)$ is visible for player~$i$.
\end{compactitem}

\paragraph{Decision problem}
Given a game $G = \tuple{Q, q_0, \delta}$ and an objective $\alpha \subseteq Q^{\omega}$,
the \emph{three-player decision problem} is to 
decide if $\exists \straa_1 \in \Straa_1 \cdot \forall \straa_2 \in \Straa_2 \cdot
\exists \straa_3 \in \Straa_3: \rho^{\straa_1,\straa_2,\straa_3}_{q_0} \in \alpha$.

\subsection{Related models}

The results for the three-player decision problem have tight connections 
and implications for decision problems on alternating tree automata and
partial-observation stochastic games that we formally define below.

\paragraph{Trees}
%We follow some definitions and notation of~\cite{KVW00}. 
%Given a finite
%sequence $w = s_0 \dots s_n \in \Alphabet^+$ over a finite set $\Alphabet$, 
%let $\last(w) = s_n$ be the last element of $w$.
%
An $\Alphabet$-labeled tree $(T,V)$ consists of a prefix-closed 
set $T \subseteq \nat^*$ (i.e., if $x\cdot d \in T$ with $x \in \nat^*$
and $d \in \nat$, then $x \in T$), and a mapping $V: T \to \Alphabet$ that assigns
to each node of $T$ a letter in $\Alphabet$.
Given $x \in \nat^*$ and $d \in \nat$ such that $x \cdot d \in T$, 
we call $x \cdot d$ the \emph{successor} in direction $d$ of $x$. 
%The \emph{successors} of a node $x \in T$ are the nodes $x \cdot d \in T$ such that $d \in \nat$.
%The degree $\deg(x)$ of a node $x \in T$ is the number of successors of $x$ in $T$. 
The node~$\varepsilon$ is the \emph{root} of the tree.
%The \emph{path-label} of a node $x = d_1 d_2 \dots d_n$ is the sequence
%of labels $V^*(x) = V(\varepsilon) V(d_1) V(s_1 d_2) \dots V(d_1 d_2 \dots d_n)$
%that labels the nodes on the path from the root to $x$ (hence $\last(V^*(x)) = V(x)$).
An \emph{infinite path} in $T$ is an infinite sequence $\pi = d_1 d_2 \dots$
of directions $d_i \in \nat$ such that every finite prefix of $\pi$ is a node in~$T$.
%A \emph{path} of $(T,V)$ is a prefix-closed set $\pi \subseteq T$ such that 
%for every $x \in \pi$, there exists a unique direction $d \in \nat$ 
%such that $x \cdot d \in \pi$.

\paragraph{Alternating tree automata}
Given a parameter $k \in \nat \setminus \{0\}$, we consider input trees of rank 
$k$, i.e. trees in which every node has at most $k$ successors. Let $[k] = \{0,\dots,k-1\}$,
and given a finite set $U$, let $\B^+(U)$ be the set of positive Boolean formulas 
over $U$, that is formulas built from elements in $U \cup \{\true,\false\}$ using 
the Boolean connectives $\land$ and $\lor$.
%We present a definition of alternating tree automata (see e.g.~\cite{MS87,KVW00}) 
%with the syntactic restriction that the states are associated to a fixed direction 
%in the input tree. The restriction is for the sake of simplifying the presentation, 
%and does not reduce the expressiveness of the class of automata (i.e., they recognize
%the regular languages of infinite trees with fixed finite rank).
%
%\mynote{L: notations: $\Alphabet$ for alphabet vs. set of strategies.}
%\mynote{L: notations: $A$ for automaton vs. set of actions.}
An \emph{alternating tree automaton} over alphabet $\Alphabet$ is a tuple 
$\A = \tuple{S,s_0,\delta_{\A}}$ where:
\begin{compactitem}
\item $S$ is a finite set of states with 
%\item $A_i$ is the finite alphabet of player~$i$ ($i=1,2,3$);
%\item 
$s_0 \in S$ the initial state; and 
\item $\delta_{\A}: S \times \Alphabet \to \B^+(S \times [k])$ is a transition function.
%\item $\dir: S \to \{0,\dots,k-1\}$ associates a fixed direction to each state.
\end{compactitem}
Intuitively, the automaton is executed from the initial state $s_0$ and reads
the input tree in a top-down fashion starting from the root~$\varepsilon$. In state $s$,
if $a \in \Alphabet$ is the letter that labels the current node $x$ of the input tree,
the behavior of the automaton is given by the formulas $\varphi = \delta_{\A}(s,a)$. 
The automaton chooses a \emph{satisfying assignment} of $\varphi$,
i.e. a set $Z \subseteq S \times [k]$ such that the formula $\varphi$ is satisfied when
the elements of $Z$ are replaced by $\true$, and the elements of $(S \times [k]) \setminus Z$
are replaced by $\false$. Then, for each $\tuple{s_1,d_1} \in Z$ a copy of the automaton is spawned 
in state $s_1$, and proceeds to the node $x \cdot d_1$ of the input tree. 
In particular, it requires that $x \cdot d_1$ belongs to the input tree. 
For example, if $\delta_{\A}(s,a) = (\tuple{s_1,0} \land \tuple{s_2,0}) \lor (\tuple{s_3,0} \land \tuple{s_4,1} \land \tuple{s_5,1})$,
then the automaton should either spawn two copies that process the successor
of~$x$ in direction~$0$~(i.e., the node $x \cdot 0$) and that enter the 
states~$s_1$ and~$s_2$ respectively, or spawn three copies of which one processes $x \cdot 0$
and enters state~$s_3$, and the other two process $x \cdot 1$ and enter 
the states~$s_4$ and~$s_5$ respectively. 

%In a standard definition of alternating tree automata~\cite{MS87,KVW00}, there is no fixed 
%direction associated to each state of the automaton. Rather the transition function can specify
%a direction to proceed along with each state to enter (the transition relation is
%then of the form $\delta_{\A}: S \times \Alphabet \to \B^+(S \times \{0,\dots,k-1\})$. And it is possible
%to specify several directions along with the same state, for instance $(s_1,0) \land (s_1,1)$
%requires that the automaton spawn two copies in state $s_1$, one that proceeds
%direction $0$ in the input tree, and one that proceeds direction $1$.
%Hence our definition can be viewed as a syntactic restriction of the standard
%definition. However, the two definitions are equally powerful as alternating tree 
%automata of the standard definition can be encoded in our definition as follows.
%For each state $s$, construct $k$ copies $(s,0), (s,1), \dots, (s,k-1)$ of $s$
%(i.e., the transition relation in each copy is the same as in $s$),
%and assign direction $\dir(s,d) = d$ for each $0 \leq d < k$.

\paragraph{Language and emptiness problem}
A run of $\A$ over a $\Alphabet$-labeled input tree $(T,V)$ is a tree $(T_r,r)$
labeled by elements of $T \times S$, where a node of $T_r$ labeled by $(x,s)$
corresponds to a copy of the automaton processing the node~$x$ of the 
input tree in state $s$. Formally, a \emph{run} of $\A$ over an input tree $(T,V)$ is a 
$(T \times S)$-labeled tree $(T_r,r)$ such that $r(\varepsilon) = (\varepsilon,s_0)$
and for all $y \in T_r$, if $r(y) = (x,s)$,
then the set $\{\tuple{s',d'} \mid \exists d \in \nat: r(y \cdot d) = (x \cdot d', s')\}$ 
is a satisfying assignment for $\delta_{\A}(s,V(x))$.
Hence we require that, given a node $y$ in $T_r$ labeled by $(x,s)$, there is a
satisfying assignment $Z \subseteq S \times [k]$ for the formula $\delta_{\A}(s,a)$ where
$a = V(x)$ is the letter labeling the current node $x$ of the input tree, 
and for all states $\tuple{s',d'} \in Z$ there is a (successor) node $y \cdot d$ in 
$T_r$ labeled by $(x \cdot d', s')$.

Given an accepting condition $\varphi \subseteq S^{\omega}$, we say that a run $(T_r,r)$
is \emph{accepting} if for all infinite paths $d_1 d_2 \dots$ of $T_r$, 
the sequence $s_1 s_2 \dots$ such that $r(d_i) = (\cdot,s_i)$ for all $i\geq 0$
is in $\varphi$. %\mynote{L: add "from the root" if we consider non-tail objectives}
The \emph{language} of $\A$ is the set $L_k(\A)$ of all input trees of rank $k$ 
over which there exists an accepting run of $\A$. The emptiness problem for alternating
tree automata is to decide, given $\A$ and parameter $k$, whether $L_k(\A) = \emptyset$.
For details related to alternating tree automata and the emptiness problem 
see~\cite{EJ91,MS87}.

%%%%%  -------------------------------------------------------

\paragraph{Two-player partial-observation stochastic games}
Given alphabet $A_i$ of actions,
and set $\Obs_i$ of observations (for player~$i \in \{1,2\}$),
a \emph{two-player partial-observation stochastic game}
(for brevity, two-player stochastic game)
is a tuple $G=\tuple{Q, q_0, \delta}$ 
where $Q$ is a finite set of states, 
$q_0 \in Q$ is the initial state, and
$\delta: Q \times A_1 \times A_2 \to \dist(Q)$ is a 
probabilistic transition where $\dist(Q)$ is the set of 
probability distributions $\kappa: Q \to [0,1]$ on $Q$,
such that $\sum_{q \in Q} \kappa(q) = 1$. 
Given a current state $q$ and
actions~$a,b$ for the players, the transition probability 
to a successor state~$q'$ is $\delta(q,a,b)(q')$.

Observation-based strategies are defined as for three-player games. 
An \emph{outcome play} from a state $q_0$ under strategies 
$\straa_1, \straa_2$ is an infinite sequence $\rho=q_0 \,a_0 b_0\, q_1 \ldots$ 
such that $a_i = \straa_1(q_0 \dots q_i)$, 
$b_i = \straa_2(q_0 \dots q_i)$, and $\delta(q_i, a_i, b_i)(q_{i+1}) > 0$
for all $i \geq 0$.

\paragraph{Qualitative analysis}
Given an objective $\alpha$ that is Borel measurable (all Borel sets in the 
Cantor topology and all objectives considered in this paper are measurable~\cite{Kechris}), a strategy $\straa_1$ for player~$1$  
is \emph{almost-sure winning} (resp., \emph{positive winning}) for the objective 
$\alpha$ from $q_0$ if for all observation-based strategies $\straa_2$ for player~$2$,
we have $\Prb_{q_0}^{\straa_1, \straa_2}(\alpha)=1$ (resp.,  $\Prb_{q_0}^{\straa_1, \straa_2}(\alpha)>0$)
where $\Prb_{q_0}^{\straa_1, \straa_2}(\cdot)$ is the unique probability measure
induced by the natural probability measure on finite prefixes of plays (i.e.,
the product of the transition probabilities in the prefix).

%\newpage

%\subsection{$\exists_{\blind} \cdot \forall_{\partial} \cdot  \exists_{\perfect}$}
%\section{$\exists_{\partial} \cdot \forall_{\partial} \cdot  \exists_{\perfect}$}

\section{Three-Player Games with Player 1 Less Informed}\label{sec:player-one-less}

We consider the three-player (non-stochastic) games defined in Section~\ref{sec:three-players}.
We show that for reachability and parity objectives the three-player decision problem 
is decidable when player~$1$ is less informed than player~$2$. 
%The complexity of this problem ranges from EXPTIME-complete when player~$2$ has
%perfect information, to EXPSPACE-complete when player~$1$ is blind, and
%2-EXPTIME-complete in general.
The problem is EXPSPACE-complete when player~$1$ is blind, and
2-EXPTIME-complete in general.

% Preliminary definition. In a game~$G$, given $s \subseteq Q$, $a \in A_1$, $b \in A_2$, let
% $\post^G(s,a,b,-) = \{q' \in Q \mid \exists q \in s \cdot \exists c \in A_3: q' = \delta(q,a,b,c)\}$.

\begin{remark}\label{rem:final-player}
Observe that for three-player (non-stochastic) games, once the strategies of 
the first two players are fixed we obtain a graph, and in graphs perfect-information 
coincides with blind for construction of a path (see~\cite[Lemma~2]{CD10b} that counting 
strategies that count the number of steps are sufficient which can be ensured by a 
player with no information).
Hence without loss of generality we consider that player~3 has 
perfect observation, and drop the observation for player~3.
\end{remark}

Our results for upper bounds are obtained by a reduction of the three-player game problem
to an equivalent exponential-size partial-observation game with only two players,
which is known to be solvable in EXPTIME for parity objectives,
and in PSAPCE when player~$1$ is blind~\cite{CDHR07}.
Our reduction preserves the number of observations of player~$1$ (thus 
if player~$1$ is blind in the three-player game, then player~$1$ is
also blind in the constructed two-player game).
Hence, the 2-EXPTIME and EXPSPACE bounds follow from this reduction.
%,and moreover it turns out that when player~$2$ has
%perfect information the reduction is polynomial, 
%giving an EXPTIME bound in that case.

\begin{theorem}[Upper bounds]\label{theo:one-less-informed-upper-bound}
Given a three-player game $G = \tuple{Q, q_0, \delta}$ with player~$1$ less 
informed than player~$2$ and a parity objective $\alpha$, 
the problem of deciding whether $\exists \straa_1 \in \Straa_1 \cdot \forall \straa_2 \in \Straa_2 \cdot
\exists \straa_3 \in \Straa_3: \rho^{\straa_1,\straa_2,\straa_3}_{q_0} \in \alpha$
can be solved in 2-EXPTIME. If player~$1$ is blind, then the problem can be solved in EXPSPACE.
%%If player~$2$ has perfect information, then the problem can be solved in EXPTIME.
%%\mynote{Mention "three-player decision problem" ?} 
\end{theorem}

\begin{proof}%[sketch]
The proof is by a reduction of the decision problem for three-player games 
to a decision problem for partial-observation two-player games with the 
same objective. 
We present the reduction for parity objectives that are visible for player~$2$ 
(defined by priority functions that are constant over observations of player~$2$).
The general case of not necessarily visible parity objectives can be solved
using a reduction to visible objectives, as in~\cite[Section 3]{CD10b}.

Given a three-player game $G = \tuple{Q, q_0, \delta}$ 
over alphabet of actions $A_i$ ($i=1,2,3$), and observations $\Obs_1, \Obs_2 \subseteq 2^Q$
for player~$1$ and player~$2$, with player~$1$ less informed than player~$2$,
we construct a two-player game $H = \tuple{Q_H, \{q_0\}, \delta_H}$ 
over alphabet of actions $A'_i$ ($i=1,2$), and observations $\Obs'_1 \subseteq 2^{Q_H}$
and perfect observation for player~2, where (intuitive explanations follow):
\begin{compactitem}
\item $Q_H = \{s \in 2^Q \mid s \neq \emptyset \land \exists o_2 \in \Obs_2: s \subseteq o_2\}$;
\item $A'_1 = A_1 \times (2^Q \times A_2 \to \Obs_2)$, and $A'_2 = A_2$;
\item $\Obs'_1 = \big\{ \{s \in Q_H \mid s \subseteq o_1\} \mid o_1 \in \Obs_1 \big\}$,
and let $\obs'_1: Q_H \to \Obs'_1$ be the corresponding observation function;
\item $\delta_H(s, (a_1,f), a_2) = \post^G(s,a_1,a_2,-) \cap f(s,a_2)$.
\end{compactitem}

Intuitively, the state space $Q_H$ is the set of knowledges of player~$2$ about the 
current state in $G$,
i.e., the sets of states compatible with an observation of player~$2$. 
Along a play in $H$, the knowledge of player~$2$ is updated to represent the set of 
possible current states in which the game~$G$ can be. In~$H$ player~$2$ has 
perfect observation and the role of player~$1$ in the game $H$ is to simulate 
the actions of both player~$1$ and player~$3$ in $G$. 
Since player~$2$ fixes his strategy before 
player~$3$ in $G$, the simulation should not let player~$2$ know
\text{player-$3$'s} action, but only the observation that player~$2$ 
will actually see while playing the game.
The actions of player~$1$ in $H$ are pairs $(a_1,f) \in A'_1$ where $a_1$
is a simple action of player~$1$ in $G$, and $f$ gives the observation $f(s,a_2)$
received by player~$2$ after the response of player~$3$ to the action $a_2$
of player~$2$ when the knowledge of player~$2$ is~$s$.
In~$H$, player~$1$ has partial observation, as he cannot distinguish
knowledges of player~$2$ that belong to the same observation of player~$1$ in $G$.
The transition relation updates the knowledges of player~$2$ as expected.
Note that $\abs{\Obs_1} = \abs{\Obs'_1}$, and therefore if player~$1$ is blind
in $G$ then he is blind in~$H$ as well.

Given a visible parity objective $\alpha = \Parity(p)$ where $p: Q \to \{0,1,\dots,d\}$ is 
constant over observations of player~$2$, let $\alpha' = \Parity(p')$
where $p'(s) = p(q)$ for all $q \in s$ and $s \in Q_H$. Note that the function
$p'$ is well defined since $s$ is a subset of an observation of player~$2$
and thus $p(q) = p(q')$ for all $q,q' \in s$. However, the parity 
objective $\alpha' = \Parity(p')$ may not be visible to player~$1$ in $G$.
We show that given a witness strategy in $G$ we can construct a witness strategy
in $H$ and vice-versa.
%, and the details of the strategy constructions are presented in Section~\ref{app:proof} of the appendix.
%We present the details of the strategy constructions to establish
%the correctness of the construction for the proof of Theorem~\ref{theo:one-less-informed-upper-bound}.
Let $\Straa_i$ be the set of observation-based
strategies of player~$i$ ($i=1,2,3$) in $G$, and let $\Straa'_i$ be the set of observation-based
strategies of player~$i$ ($i=1,2$) in $H$.
We claim that the following statements are equivalent:
\begin{itemize}
\item[$(1)$] In $G$, $\exists \straa_1 \in \Straa_1 \cdot \forall \straa_2 \in \Straa_2 \cdot
\exists \straa_3 \in \Straa_3: \rho^{\straa_1,\straa_2,\straa_3}_{q_0} \in \alpha$.
\item[$(2)$] In $H$, $\exists \straa'_1 \in \Straa'_1 \cdot \forall \straa'_2 \in \Straa'_2: \rho^{\straa'_1,\straa'_2}_{\{q_0\}} \in \alpha'$.
\end{itemize}

The 2-EXPTIME result of the theorem follows from this equivalence because the game $H$
is at most exponentially larger than the game $G$, and 
two-player partial-observation games with a parity objective can be 
solved in EXPTIME, and when player~$1$ is blind they can be solved in PSPACE~\cite{CDHR07}.
Observe that when player~$2$ has perfect information, his observations are singletons
and $H$ is no bigger than $G$, and an EXPTIME bound follows in that case.

To show that $(1)$ implies $(2)$, let $\straa_1: \Obs_1^+ \to A_1$ be a 
strategy for player~$1$ such that for all strategies $\straa_2: \Obs_2^+ \to A_2$,
there is a strategy $\straa_3: \Obs_3^+ \to A_3$ such that 
$\rho^{\straa_1,\straa_2,\straa_3}_{q_0} \in \alpha$. 
%If we identify an observation $o_1' \in \Obs_1'$ with the
%observation $o_1 \in \Obs_1$ such that $u \subseteq o_1$ for all $u \in o'_1$, then
%we can view $\straa_1$ as a function $\straa_1: \Obs'^+_1 \to A_1$.
%
From $\straa_1$, we construct an (infinite) DAG over state space $Q_H \times \Obs_1^+$ with edges labeled
by elements of $A_2 \times \Obs_2$ defined as follows. The root is 
$(\{q_0\},\obs_1(q_0))$. There is an edge labeled by $(b,o_2) \in A_2 \times \Obs_2$ 
from $(s,\rho)$ to $(s',\rho')$ if $s' = \post^G(s,a,b,-) \cap o_2 \neq \emptyset$
where $a = \straa_1(\rho)$, and $\rho' = \rho \cdot o_1$ where $o_1 \in \Obs_1$
is the (unique) observation of player~$1$ such that $o_2 \subseteq o_1$.
Note that for every node $n = (s,\rho)$ in the DAG, for all states $q \in s$, 
for all $b \in A_2, c \in A_3$, there is a successor $n' = (s',\rho')$ of 
$n$ such that $\delta(q,a,b,c) \in s'$ where $a = \straa_1(\rho)$. 
Consider a perfect-information turn-based game played over this DAG, 
between player~$2$ choosing actions $b \in A_2$
and player~$3$ choosing observations $o_2 \in \Obs_2$, resulting in an
infinite path $(s_0, \rho_0) (s_1,\rho_1) \dots $ in the DAG as expected, and
that is defined to be winning for player~$3$ if the sequence $s_0 s_1 \dots$
satisfies $\alpha'$. We show that in this game, for all strategies
of player~$2$ (which naturally define functions $\straa_2: \Obs_2^+ \to A_2$), 
there exists a strategy of player~$3$ (a function $f_3: Q_H \times \Obs_1^+ \times A_2 \to \Obs_2^+$)
to ensure that 
the resulting play satisfies $\alpha'$. The argument is
based on $(1)$ saying that given the strategy $\straa_1$ is fixed, 
for all strategies $\straa_2: \Obs_2^+ \to A_2$,
there is a strategy $\straa_3: \Obs_3^+ \to A_3$ such that 
$\rho^{\straa_1,\straa_2,\straa_3}_{q_0} \in \alpha$. 
Given a strategy for player~$2$ in the game over the DAG, we use
$\straa_3$ to choose observations $o_2 \in \Obs_2$ as follows.
We define a labelling function $\lambda: Q_H \times \Obs_1^+ \to Q$ over the DAG 
in a top-down fashion such that $\lambda(s,\rho) \in s$. First, let 
$\lambda(\{q_0\}, \obs_1(q_0)) = q_0$, and given 
$\lambda(s, \rho) = q$ with an edge labeled by $(b,o_2)$ to $(s',\rho')$,
let $\lambda(s', \rho') = \delta(q,a,b,c)$
where $a = \straa_1(\rho)$ and $c =  \straa_3(\rho)$.
Note that indeed $\delta(q,a,b,c) \in s'$.
Now we define a strategy for player~$3$ that,
in a node $(s, \rho)$ of the DAG, chooses the observation
$\obs_2(\delta(q,a,b,c))$ where $q = \lambda(s, \rho)$,
$a = \straa_1(\rho)$, $b$ is the action chosen by player~$2$
at that node (remember we fixed a strategy for player~$2$), 
and $c =  \straa_3(\rho)$. Since $\lambda(s,\rho) \in s$,
it follows that the resulting play satisfies $\alpha'$
since $\rho^{\straa_1,\straa_2,\straa_3}_{q_0}$ satisfies $\alpha$.

By determinacy of perfect-information turn-based games~\cite{Martin75},
in the game over the DAG there exists a strategy $f_3$ for player~$3$ such that
for all player-$2$ strategies, the outcome play satisfies $\alpha'$.
Using $f_3$, we construct a strategy $\straa'_1$ for player~$1$ in $H$ as follows. 
First, by a slight abuse of notation, we identify the observations $o_1' \in \Obs_1'$ 
with the observation $o_1 \in \Obs_1$ such that $u \subseteq o_1$ for all $u \in o'_1$.
For all $\rho \in \Obs_1^+$, let $\straa'_1(\rho) = (a,f)$ where $a = \straa_1(\rho)$ 
and $f$ is defined by $f(s,a_2) = f_3(s,\rho,a_2)$. 
By construction of the DAG and of the strategy $\straa'_1$, for all strategies
of player~$2$ in $H$ the outcome play $s_0 s_1 \dots$ satisfies the parity
objective $\alpha'$, and thus $\straa'_1$ is a winning 
observation-based strategy in $H$.

To show that $(2)$ implies $(1)$, let $\straa'_1$ be a winning 
observation-based strategy for the objective $\alpha'$ in $H$.
Consider the DAG over state space $Q_H \times \Obs_1^+$ with edges labeled
by elements of $A_2$ defined as follows. 
The root is $(\{q_0\},\obs_1(q_0))$. 
For all nodes $(s,\rho)$, for all $b \in A_2$,
there is an edge labeled by $b$ from $(s,\rho)$ to $(s',\rho')$ 
if $s' = \post^G(s,a,b,-) \cap o_2$ and $\rho' = \rho \cdot o_1$
where $o_2 = f(s,b)$ and $(a,f) = \straa'_1(\rho)$,
and $o_1 \in \Obs_1$ is the (unique) observation of player~$1$ such that 
$o_2 \subseteq o_1$.
We say that $(s',\rho')$ is the $b$-successor of $(s,\rho)$. 
Note that for all $q' \in s'$, there exists $q \in s$ and $c \in A_3$
such that $q' = \delta(q,a,b,c)$.

This DAG mimics the unraveling of $H$ under $\straa'_1$,
and since $\straa'_1$ is a winning strategy, for all infinite 
paths $(s_0, \rho_0) (s_1,\rho_1) \dots $ of the DAG, the
sequence $s_0 s_1 \dots$ satisfies $\alpha'$. 

Define the strategy $\straa_1$ such that $\straa_1(\rho) = a$ if $\straa'_1(\rho) = (a,f)$
(again identifying the observations in $\Obs_1'$ and $\Obs_1$). To show that 
$(1)$ holds, fix an arbitrary observation-based strategy $\straa_2$ for player~$2$.
The outcome play of $\straa_1$ and $\straa_2$ in $H$ is the sequence 
$(s_0,\rho_0) (s_1,\rho_1) \dots $ where $(s_0,\rho_0)$
is the root, and such that for all $i \geq 1$,
the node $(s_i,\rho_i)$ is the $b$-successor of $(s_{i-1},\rho_{i-1})$ 
where $b = \straa_2(\obs_2(s_0 s_1 \dots s_{i-1}))$ (where $\obs_2(s_i)$
is naturally defined as the unique observation $o_2 \in \Obs_2$ 
such that $s_i \subseteq o_2$). From this path in the DAG, we 
construct an infinite path $p_0 p_1 \dots$ in $G$ using K\"onig's Lemma~\cite{Konig36} as follows. 
First, it is easy to show by induction (on $k$) that for every finite prefix
$s_0 s_1 \dots s_k$ and for every $p_k \in s_k$ there exists a path
$p_0 p_1 \dots p_k$ in $G$ such that $p_i \in s_i$ for all $0 \leq i \leq k$.
Note that $p_0 = q_0$ since $s_0 = \{q_0\}$ and that by 
definition of the DAG, for each $s_{i+1}$ ($i=0,\dots,k-1$),
there exist $a \in A_1$, $b \in A_2$, and $o_2 \in \Obs_2$ such that 
$s_{i+1} = \post^G(s_{i},a,b,-) \cap o_2$. Hence, given $p_{i+1} \in s_{i+1}$,
there exist $c_i \in A_3$ and $p_{i} \in s_{i}$ such that 
$\delta(p_{i},a,b,c_i) = p_{i+1}$.

Arranging all these finite paths in a tree, we obtain an
infinite finitely-branching tree which by K\"onig's Lemma~\cite{Konig36} contains an infinite
branch $q_0 q_1 \dots$ that is a path in $G$ and such that $q_i \in s_i$ 
for all $i \geq 0$. 
Now we can construct the strategy $\straa_3$ such that $\straa_3(p_0 \dots p_i) = c_i$.
Since $s_0 s_1 \dots$ satisfies $\alpha'$, it follows that 
$\rho^{\straa_1,\straa_2,\straa_3}_{q_0} = p_0 p_1 \dots$ satisfies $\alpha$,
which completes the proof.
\qed
\end{proof}

%%We proceed with a matching lower bound for the three-player decision problem.

\begin{theorem}[Lower bounds]\label{theo:one-less-informed-lower-bound}
Given a three-player game $G = \tuple{Q, q_0, \delta}$ with player~$1$ less informed 
than player~$2$ and a reachability objective $\alpha$, 
the problem of deciding whether $\exists \straa_1 \in \Straa_1 \cdot \forall \straa_2 \in \Straa_2 \cdot
\exists \straa_3 \in \Straa_3: \rho^{\straa_1,\straa_2,\straa_3}_{q_0} \in \alpha$
is 2-EXPTIME-hard. 
If player~$1$ is blind (and even when player~2 is also blind), then the problem is EXPSPACE-hard.
%\mynote{L: mention hardness when player~$2$ has perfect information ?}
%\mynote{Mention "three-player decision problem" ?}
\end{theorem}

\begin{proof}%[sketch]
The proof of 2-EXPTIME-hardness is obtained by a polynomial-time reduction of the membership problem
for exponential-space \emph{alternating} Turing machines to the three-player problem. The same reduction for
the special case of exponential-space \emph{nondeterministic} Turing machines shows
EXPSPACE-hardness when player~$1$ is blind (because our reduction yields
a game in which player~$1$ is blind when we start from a nondeterministic Turing
machine). 
The membership problem for Turing machines is to decide, given a Turing machine $M$
and a finite word $w$, whether $M$ accepts $w$. 
The membership problem is 2-EXPTIME-complete for exponential-space alternating 
Turing machines, and EXPSPACE-complete for exponential-space nondeterministic 
Turing machines~\cite{papa-book}.

An alternating Turing
machine is a tuple $M = \tuple{Q_{\lor},Q_{\land},\Sigma, \Gamma, \Delta, q_0, q_{acc}, q_{rej}}$
where the state space $Q = Q_{\lor} \cup Q_{\land}$ consists of the set $Q_{\lor}$ of or-states,
and the set $Q_{\land}$ of and-states. The input alphabet is $\Sigma$,
the tape alphabet is $\Gamma = \Sigma \cup \{\#\}$ where $\#$ is the blank symbol.
The initial state is $q_0$, the accepting state is $q_{acc}$, and the rejecting state is $q_{rej}$. 
The transition relation is $\Delta \subseteq Q \times \Gamma \times Q \times \Gamma \times \{-1,1\}$,
where a transition $(q,\gamma,q',\gamma',d) \in \Delta$ intuitively means that,
given the machine is in state $q$, and the symbol under the tape head is $\gamma$, 
the machine can move to state $q'$, replace the symbol under the tape head 
by $\gamma'$, and move the tape head to the neighbor cell in direction~$d$.
A configuration~$c$ of~$M$ is a sequence $c \in (\Gamma \cup (Q \times \Gamma))^{\omega}$
with exactly one symbol in $Q \times \Gamma$, which indicates the current state 
of the machine and the position of the tape head.
The initial configuration of $M$ on $w = a_0 a_1 \dots a_n$ 
is $c_0 = (q_0,a_0) \cdot a_1 \cdot a_2 \cdot \dots \cdot a_n \cdot \#^{\omega}$.
Given the initial configuration of $M$ on $w$,
it is routine to define the execution trees of $M$ where at least one successor 
of each configuration in an or-state, and all successors of 
the configurations in an and-state are present (and we assume that all
branches reach either $q_{acc}$ or $q_{rej}$), and to say that $M$ accepts
$w$ if all branches of some execution tree reach $q_{acc}$. Note that $Q_{\land} = \emptyset$ 
for nondeterministic Turing machines, and in that case the execution tree reduces to a single path.
A Turing machine $M$ uses exponential space if for all words $w$, 
all configurations in the execution of $M$ on $w$ contain at most 
$2^{O(\abs{w})}$ non-blank symbols.

We present the key steps of our reduction from alternating Turing machines. Given a Turing machine
$M$ and a word $w$, we construct a three-player game with reachability objective
in which player~$1$ and player~$2$ have to simulate the execution of $M$ on $w$,
and player~$1$ has to announce the successive configurations and transitions of the machine along 
the execution. Player~$1$ announces configurations one symbol at a time, thus the alphabet
of player~$1$ is $A_1 = \Gamma \cup (Q \times \Gamma) \cup \Delta$. 
In an initialization phase, the transition
relation of the game forces player~$1$ to announce the initial configuration  
$c_0$ (this can be done with $O(n)$ states in the game, where $n = \abs{w}$).
Then, the game proceeds to a loop where player~$1$ keeps announcing symbols of
configurations. At all times along the execution, some finite information is stored
in the finite state space of the game: a window of the last three symbols $\z_{1}, \z_{2}, \z_3$
announced by player~$1$, as well as the last symbol $\head \in Q \times \Gamma$ 
announced by player~$1$ (that indicates the current machine state and the position
of the tape head). After the initialization phase, we should have 
$\z_{1} = \z_{2} = \z_3 = \#$ and $\head = (q_0,a_0)$. When player~$1$ has 
announced a full configuration, he moves to a state of the game where either player~$1$
or player~$2$ has to announce a transition of the machine: for $\head = (p,a)$,
if $p \in Q_{\lor}$, then player~$1$ chooses the next transition, and if 
$p \in Q_{\land}$, then player~$2$ chooses. Note that the transitions chosen
by player~$2$ are visible to player~$1$ and this is the only information
that player~$1$ observes. Hence player~$1$ is less informed than player~$2$, and
both player~$1$ and player~$2$ are blind when the machine is nondeterministic.
If a transition $(q,\gamma,q',\gamma',d)$
is chosen by player~$i$, and either $p \neq q$ or $a \neq \gamma$, then player~$i$
loses (i.e., a sink state is reached to let player~$1$ lose, and the target state 
of the reachability objective is reached to let player~$2$ lose).
If at some point player~$1$ announces a symbol $(p,a)$ with $p = q_{acc}$,
then player~$1$ wins the game.

The role of player~$2$ is to check that player~$1$ faithfully simulates the 
execution of the Turing machine, and correctly announces the configurations.
After every announcement of a symbol by player~$1$, the game offers the possibility
to player~$2$ to compare this symbol with the symbol at the same position in the next configuration.
We say that player~$2$ \emph{checks} (and whether player~$2$ checks or not is not
visible to player~$1$), and the checked symbol is stored as $\z_{2}$. 
Note that player~$2$ can be blind to check because player~$2$ fixes his strategy
after player~$1$.
The window $\z_{1}, \z_{2}, \z_3$ stored in the
state space of the game provides enough information to update the middle cell $\z_{2}$
in the next configuration, and it allows the game to verify the check of player~$2$.
However, the distance (in number of steps) between the same position in two 
consecutive configurations is exponential (say $2^n$ for simplicity), 
and the state space of the game is not large enough to check
that such a distance exists between the two symbols compared by player~$2$.
We use player~$3$ to check that player~$2$ makes a comparison at the correct position. 
When player~$2$ decides to check, he has to count from $0$ to $2^n$ by announcing after 
every symbol of player~$1$ a sequence of $n$ bits, initially all zeros (again,
this can be enforced by the structure of the game with  $O(n)$ states). It is
then the responsibility of player~$3$ to check that player~$2$ counts correctly.
To check this, player~$3$ can at any time choose a bit position $p \in \{0,\dots,n-1\}$ and 
store the bit value~$b_p$ announced by player~$2$ at position~$p$. 
The value of $b_p$ and $p$ is not visible to player~$2$. While player~$2$ announces
the bits $b_{p+1}, \dots, b_{n-1}$ at position $p+1, \dots, n-1$, 
the finite state of the game is used to flip the value of $b_p$ if all bits 
$b_{p+1}, \dots, b_{n-1}$ are equal to $1$, hence updating $b_p$ to the value 
of the $p$-th bit in what should be the next announcement of player~$2$.
In the next bit sequence announced by player~$2$, the $p$-th bit is compared with~$b_p$. If they match, then 
the game goes to a sink state (as player~$2$ has faithfully counted), and if they differ
then the game goes to the target state (as player~$2$ is caught cheating).
It can be shown that this can be enforced by the structure of the game with $O(n^2)$ states,
that is $O(n)$ states for each value of $p$. As before, whether player~$3$
checks or not is not visible to player~$2$.

Note that the checks of player~$2$ and player~$3$ are one-shot: the game will
be over (either in a sink or target state) when the check is finished. This is
enough to ensure a faithful simulation by player~$1$, and a faithful counting
by player~$2$, because $(1)$ partial observation allows to hide to a player
the time when a check occurs, and $(2)$ player~$2$ fixes his strategy
after player~$1$ (and player~$3$ after player~$2$), thus they can decide 
to run a check exactly when player~$1$ (or player~$2$) is not faithful. 
This ensures that player~$1$ does not win if he does not simulate the execution 
of $M$ on $w$, and that player~$2$ does not win if he does not count correctly.

Hence this reduction ensures that $M$ accepts $w$ if and only if the answer to 
the three-player game problem is {\sc Yes}, where the reachability objective is
satisfied if player~$1$ eventually announces
that the machine has reached $q_{acc}$ (that is if $M$ accepts $w$), or
if player~$2$ cheats in counting, which can be detected by player~$3$.
\qed
\end{proof}

\section{Three-Player Games with Player 1 Perfect}\label{sec:player-one-perfect}

When player 2 is less informed than player 1, we show that three-player games 
get much more complicated (even in the special case where player~$1$ has
perfect information). We note that for reachability objectives, 
the three-player decision problem is equivalent to the qualitative analysis 
of positive winning in two-player stochastic games, and we show that the techniques
developed in the analysis of two-player stochastic games can be extended
to solve the three-player decision problem with safety objectives as well.

%%\mynote{L: note that what follows is for general three-player games, not just for player~$1$ perfect.}

%\begin{remark}\label{rem:eq}
For reachability objectives, the three-player decision problem is equivalent
to the problem of positive winning in two-player stochastic games where the 
third player is replaced by a probabilistic choice over the action
set with uniform probability. Intuitively, after player~$1$ and player~$2$
fixes their strategy, the fact that player~$3$ can construct a (finite)
path to the target set is equivalent to the fact that such a path
has positive probability when the choices of player~$3$ are replaced 
by uniform probabilistic transitions.
%\end{remark}%
%
Given a three-player game $G = \tuple{Q, q_0, \delta}$,
let $\Uniform(G) = \tuple{Q, q_0, \delta'}$ be the two-player partial-observation 
\emph{stochastic} game (with same state space, action sets, and observations for player~$1$ 
and player~$2$) where $\delta'(q,a_1,a_2)(q') = 
\frac{\left\lvert\{a_3 \mid \delta(q,a_1,a_2,a_3)=q'\}\right\rvert}{\abs{A_3}}$
for all $a_1 \in A_1$, $a_2 \in A_2$, and $q,q' \in Q$.
Formally, the equivalence result is presented in Lemma~\ref{lem:uniform}, and the 
equivalence holds for all three-player games (not restricted to three-player games
where player~1 has perfect information).
However, we will use Lemma~\ref{lem:uniform} to establish results for three-player 
games where player~1 has perfect information.

\begin{lemma}\label{lem:uniform}
Given a three-player game $G$ and a reachability objective $\alpha$, 
the answer to the three-player decision problem for $\tuple{G,\alpha}$ is {\sc Yes} 
if and only if player~$1$ is positive winning for $\alpha$ in 
the two-player partial-observation stochastic game $\Uniform(G)$.
\end{lemma}

%\mynote{L: should we make clear that Lemma~\ref{lem:uniform} holds in general 
%(not only for player~$1$ perfect ?}
%It follows from Lemma~\ref{lem:uniform} 
%that the result of Theorem~\ref{theo:one-less-informed-upper-bound} 
%generalizes the complexity result of~\cite[Theorem~1]{CD12},
%which established EXPTIME-completeness in stochastic two-player 
%reachability games
%with player~$2$ having perfect information. In particular, when
%both player~$1$ and player~$2$ have partial observation,
%Theorem~\ref{theo:one-less-informed-upper-bound} can be used to
%show that two-player stochastic games with reachability objective
%can be solved in 2-EXPTIME when player~$1$ is less informed than player~$2$,
%extending the results of~\cite{CD12}.

\paragraph{Reachability objectives.}
Even in the special case where player~$1$ has
perfect information, and for reachability objectives, non-elementary
memory is necessary in general for player~$1$ to win in three-player games. 
This result follows from Lemma~\ref{lem:uniform}
%established in Section~\ref{sec:player-one-less},
and from the result of~\cite[Example 4.2 Journal version]{CD12} showing that   % ~\cite[Theorem~2]{CD12}
non-elementary memory is necessary to win with positive probability
in two-player stochastic games.
It also follows from Lemma~\ref{lem:uniform} and the result
of~\cite[Corollary 4.9 Journal version]{CD12} that the three-player decision problem % ~\cite[Corollary~1]{CD12}
for reachability games is decidable in non-elementary time.
%We extend the decidability result to safety objectives (see Section~\ref{app:safety}
%in Appendix).

\paragraph{Safety objectives.}
We show that the three-player decision problem can be solved for games
with a safety objective when player~$1$ has perfect information.
The proof is using the \emph{counting abstraction} of~\cite[Section 4.2 Journal version]{CD12}    % ~\cite[Section IV]{CD12} \url{http://tocl.acm.org/accepted/chatterjee-doyen_partial.pdf}    
and shows that the answer to the three-player decision problem for safety objective $\Safe(\target)$
is {\sc Yes} if and 
only if there exists a winning strategy in the two-player counting-abstraction game
with the safety objective to visit only counting functions (i.e., essentially
tuples of natural numbers) with support contained in the target states $\target$.
Intuitively, the counting abstraction is as follows: with every knowledge of player~2
we store a tuple of counters, one for each state in the knowledge.
The counters denote the number of possible distinct paths to the states of the knowledge, 
and the abstraction treats large enough values as infinite (value $\omega$).
The counting-abstraction game is monotone with regards to the natural
partial order over counting functions, and therefore it is well-structured
and can be solved by constructing a self-covering unraveling tree, i.e. a tree
in which the successors of a node are constructed only if this node
has no greater ancestor. The properties of well-structured systems 
(well-quasi-ordering and K\"onig's Lemma)
ensure that this tree is finite, and that there exists a strategy to ensure only 
supports contained in the target states $\target$ are visited if and only if
there exists a winning strategy in the counting-abstraction game (in a leaf
of the tree, one can copy the strategy played in a greater ancestor).
It follows that the three-player decision problem for safety games
is equivalent the problem of solving a safety game over this finite tree.

\begin{theorem}\label{theo:player-one-perfect}
When player 1 has perfect information, the three-player decision problem is decidable for both
reachability and safety games, and for reachability games 
memory of size non-elementary is necessary in general for player~$1$.
\end{theorem}

\section{Four-Player Games}\label{sec:more-than-three}

We show that the results presented for three-player games
extend to games with four players (the fourth player is universal 
and perfectly informed). The definition of four-player games
and related notions is a straightforward extension of Section~\ref{sec:three-players}.

%\mynote{L: definition of alternating tree automata.}
In a four-player game with player~$1$ less informed than player~$2$, and perfect information
for both player~$3$ and player~$4$, consider the \emph{four-player decision problem}
which is to decide if 
$\exists \straa_1 \in \Straa_1 \cdot \forall \straa_2 \in \Straa_2 \cdot
\exists \straa_3 \in \Straa_3 \cdot \forall \straa_4 \in \Straa_4: \rho^{\straa_1,\straa_2,\straa_3,\straa_4}_{q_0} \in \alpha$
for a parity objective $\alpha$.
Since player~$3$ and player~$4$ have perfect information, we assume without loss
of generality that the game is turn-based for them, that is there
is a partition of the state space $Q$ into two sets $Q_3$ and $Q_4$ (where
$Q = Q_3 \cup Q_4$) such that the transition function is the union
of $\delta_3: Q_3 \times A_1 \times A_2 \times A_3 \to Q$ and
$\delta_4: Q_4 \times A_1 \times A_2 \times A_4 \to Q$.
Strategies and outcomes are defined analogously to three-player games.
A strategy of player~$i \in \{3,4\}$ is of the form 
$\straa_i: Q^* \cdot Q_i \to A_i$.

By determinacy of perfect-information turn-based games with countable state space~\cite{Martin75},
the negation of the four-player decision problem is equivalent to $\forall \straa_1 \in \Straa_1 \cdot \exists \straa_2 \in \Straa_2 \cdot
\exists \straa_4 \in \Straa_4 \cdot \forall \straa_3 \in \Straa_3: 
\rho^{\straa_1,\straa_2,\straa_3,\straa_4}_{q_0} \in \alpha$.
Once the strategies $\straa_1$ and $\straa_2$ are fixed, the condition 
$\exists \straa_4 \in \Straa_4 \cdot \forall \straa_3 \in \Straa_3: 
\rho^{\straa_1,\straa_2,\straa_3,\straa_4}_{q_0} \in \alpha$ can be viewed
as the membership problem for a tree $t^{\straa_1,\straa_2}$
in the language of an alternating parity tree automaton~\cite{CDNV14}
with state space~$Q$
where $t^{\straa_1,\straa_2}$ is the $(A_1 \times A_2)$-labeled tree $(T,V)$
where $T = \Obs_2^+$ and $V(\rho) = (\straa_1(\obs_1(\rho)), \straa_2(\rho))$ for all $\rho \in T$.

By the results of~\cite{EJ91}, if there exists an accepting $(\Obs_2^+ \times Q)$-labeled run tree $(T_r, r)$
for an input tree $t^{\straa_1,\straa_2}$ in an alternating parity tree automaton,
then there exists a \emph{memoryless} accepting run tree, that is such that 
for all nodes $x,y \in T_r$ such that $\abs{x} = \abs{y}$ and $r(x) = r(y)$,
the subtrees of $T_r$ rooted at $x$ and $y$ are isomorphic.
Since the membership problem is equivalent to a two-player parity game played
on the structure of the alternating automaton, a memoryless accepting run tree 
can be viewed as a winning strategy $\straa_4: \Obs_2^+ \times Q \to A_4$,
or equivalently $\straa_4: \Obs_2^+ \to (Q \to A_4)$ such that for all
strategies $\straa_3: T_r \to A_3$, the resulting infinite branch in the tree $T_r$
satisfies the parity objective $\alpha$.

It follows from this that the (negation of the) original question 
$\forall \straa_1 \in \Straa_1 \cdot \exists \straa_2 \in \Straa_2 \cdot
\exists \straa_4 \in \Straa_4 \cdot \forall \straa_3 \in \Straa_3: 
\rho^{\straa_1,\straa_2,\straa_3,\straa_4}_{q_0} \in \alpha$ is equivalent
to $\forall \straa_1 \in \Straa_1 \cdot \exists \straa_{24} \in \Straa_{24} \cdot
\forall \straa_3 \in \Straa_3: \rho^{\straa_1,\straa_{24},\straa_3}_{q_0} \in \alpha$
where $\Straa_{24} = \Obs_2^+ \to (A_2 \times (Q \to A_4))$ is the set of strategies
of a player (call it player 24) with observations $\Obs_2$ and action set $A_2' = A_2 \times (Q \to A_4)$,
and the outcome $\rho^{\straa_1,\straa_{24},\straa_3}_{q_0}$ is defined
as expected in a three-player game (played by player~$1$, player~$24$, and player~$3$)
with transition function $\delta': Q \times A_1 \times A_2' \times A_3 \to Q$
defined by $\delta'(q,a_1,(a_2,f),a_3) = \delta(q,a_1,a_2,a_3,f(q))$.

Hence the original question (and its negation) for four-player games 
reduces in polynomial time to solving a three-player game with the first player less 
informed than the second player. Hardness follows from the special case
of three-player games.

\begin{theorem}\label{theo:player-four}
The four-player decision problem with player~$1$ less informed than player~$2$,
and perfect information for both player~$3$ and player~$4$
is 2-EXPTIME-complete for parity objectives.
\end{theorem}

\begin{remark}[Combinations of strategy quantification.]
We now discuss the various possibilities of strategy quantifiers and 
information of the players in multi-player games.
First, if there are two existential (resp., universal) players 
with incomparable information, then the decision question is 
undecidable~\cite{RP79,PR89}; and if there is a sequence of existential 
(resp., universal) quantification over strategies players such that 
the information of the players form a chain (i.e., in the sequence
of quantification over the players, let the players be~$i_1,i_2,\ldots,i_k$
such that $i_1$ is more informed than $i_2$, $i_2$ more informed than $i_3$
and so on), then with repeated subset construction, the sequence can 
be reduced to one quantification~\cite{PR89,MT01,MW03}. 
Note however that if there is a quantifier alternation between existential
and universal, then even if the information may form a chain, subset 
construction might not be sufficient: for example, if 
player~1 is perfect and player~2 has partial-information, non-elementary 
memory might be necessary (as shown in Section~\ref{sec:player-one-perfect}).
We now discuss the various possibilities of strategy quantification 
in four-player games.
Without loss of generality we consider that the first strategy quantifier is 
existential.
The above argument for sequence of quantifiers (either undecidability 
with incomparable information or the sequence reduces to one) shows that we 
only need to consider the following strategy quantification:
$\exists_1 \forall_2 \exists_3 \forall_4$, where the subscripts denote the 
quantification over strategies for the respective player.
First, note that once the strategies of the first three players are fixed
we obtain a graph, and similar to Remark~\ref{rem:final-player} 
%in graphs perfect-information coincides with blind 
%for construction of a path (see~\cite[Lemma~2]{CD10b} that counting strategies
%are sufficient).
%%Hence 
without loss of generality we consider that player~4 has 
perfect observation.
We now consider the possible cases for player~3 in presence of player~4.
\begin{compactenum}
\item \emph{Perfect observation.} The case when player~3 has perfect 
observation has been solved in the main paper (results of 
Section~\ref{sec:more-than-three}).

\item \emph{Partial observation.} We now consider the case when player~3 has 
partial observation. 
If player~2 is less informed than player~1, then the problem is at least as 
hard as the problem considered in Section~\ref{sec:player-one-perfect}. 
If player~3 is less informed than player~2, then even in the absence of player~1,
the problem is as hard as the negation of the question considered in Section~\ref{sec:player-one-perfect}
(where first a more informed player plays, followed by a less informed player,
just the strategy quantifiers are  $\forall_2 \exists_3 \forall_4$ as 
compared to $\exists_1 \forall_2 \exists_3$ considered in Section~\ref{sec:player-one-perfect}).
Finally, if player~1 is less informed than player~2, and player~2 is less informed
than player~3, then we apply our construction of Section~\ref{sec:player-one-less} twice and 
obtain a double exponential size two-player partial-observation game which can be solved
in 3-EXPTIME.
Recall that in absence of player~4, by Remark~\ref{rem:final-player} whether player~3 has 
partial or perfect information does not matter and we obtain a 2-EXPTIME upper bound; 
whereas in presence of player~4, we obtain a 3-EXPTIME upper bound
if player~3 has partial information (but more informed than player~$2$), 
and a 2-EXPTIME upper bound if player~3 has perfect information 
(Theorem~\ref{theo:player-four}).
\end{compactenum}
\end{remark}

%%\mynote{Remark about strategy quantifiers}

%\input{applications}

\section{Applications}\label{sec:applications}

We now discuss applications of our results in the context
of synthesis and qualitative analysis of two-player partial-observation 
stochastic games.

\smallskip\noindent{\bf Sequential synthesis.}
%\subsection{Sequential synthesis}
The \emph{sequential synthesis} problem consists of an open system of partially 
implemented modules (with possible non-determinism or choices) $M_1, M_2, 
\ldots, M_n$ that need to be refined (i.e., the choices determined by 
strategies) such that the composite system after refinement satisfy a 
specification. 
The system is open in the sense that after the refinement the composite system 
is reactive and interact with an environment.
Consider the problem where first a set $M_1, \ldots, M_k$ of modules are 
refined, then a set $M_{k+1}, \ldots,M_{\ell}$ are refined by an external 
implementor, and finally the remaining set of modules are refined. 
In other words, the modules are refined sequentially: first a set of modules 
whose refinement can be controlled, then a set of modules whose refinement 
cannot be controlled as they are implemented externally, and finally the 
remaining set of modules. 
If the refinements of modules $M_1, \ldots, M_{\ell}$ do not have access to 
private variables of the remaining modules we obtain a partial-observation 
game with four players: the first (existential) player corresponds to the 
refinement of modules $M_1,\ldots,M_k$, the second (universal) player 
corresponds to the refinement of modules $M_{k+1},\ldots,M_{\ell}$, the third 
(existential) player corresponds to the refinement of the remaining modules, 
and the fourth (adversarial) player is the environment. 
If the second player has access to all the variables visible to the first player, 
then player~1 is less informed.

\smallskip\noindent{\bf Two-player partial-observation stochastic games.}
%\subsection{Two-player partial-observation stochastic games}
Our results for four-player games imply new complexity results for 
two-player stochastic games. 
For qualitative analysis (positive and almost-sure winning) under finite-memory 
strategies for the players the following reduction has been established 
in~\cite[Lemma~1]{CDNV14} (see Lemma~2.1 of the arxiv version): 
the probabilistic transition function can be replaced by a 
turn-based gadget consisting of two perfect-observation players, one angelic 
(existential) and one demonic (universal).
The turn-based gadget is the same as used for perfect-observation stochastic 
games~\cite{Cha-Thesis,CJH03}.
In~\cite{CDNV14}, only the special case of perfect observation for player~2 
was considered, and hence the problem reduced to three-player games where only 
player~1 has partial observation and the other two players have perfect 
observation.
In case where player~2 has partial observation, the reduction of~\cite{CDNV14}
requires two perfect-observation players, and gives the problem of four-player 
games (with perfect observation for player~3 and player~4).
Hence when player~1 is less informed, we obtain a 2-EXPTIME upper bound from 
Theorem~\ref{theo:player-four}, and obtain a 2-EXPTIME lower 
bound from Theorem~\ref{theo:one-less-informed-lower-bound} since 
the three-player games problem with player~1 less informed for reachability 
objectives coincides with \emph{positive} winning for two-player partial-observation 
stochastic games (Lemma~\ref{lem:uniform}). 

For \emph{almost-sure} winning, a 2-EXPTIME lower bound can also 
be obtained by an adaptation of the proof of Theorem~\ref{theo:one-less-informed-lower-bound}.
We use the same reduction from exponential-space alternating Turing machines,
with the following changes: $(i)$ the third player is replaced by a 
uniform probability distribution over player-$3$'s moves, thus the reduction
is now to two-player partial-observation stochastic games; $(ii)$
instead of reaching a sink state when player~$2$ 
detects a mistake in the sequence of configurations announced by player~$1$,
the game restarts in the initial state; thus the target state of the reachability
objective is not reached, but player~$1$ gets another chance to faithfully
simulate the Turing machine.

It follows that if the Turing machine accepts, then player~$1$ has an almost-sure
winning strategy by faithfully simulating the execution. Indeed, either $(a)$ player~$2$
never checks, or checks and counts correctly, and then player~$1$ wins since
no mistake is detected, or $(b)$ player~$2$ checks and cheats counting, and then
player~$2$ is caught with positive probability (player~$1$ wins), and with probability
smaller than~$1$ the counting cheat is not detected and thus possibly a (fake) 
mismatch in the symbol announced by player~$1$ is detected. Then the game
restarts. Hence in all cases after finitely many steps, either player~$1$ wins with (fixed) 
positive probability, or the game restarts. It follows that player~$1$ wins
the game with probability~$1$.

If the Turing machine rejects, then player~$1$ cannot win by a faithful
simulation of the execution, and thus he should cheat. The strategy 
of player~$2$ is then to check and to count correctly, ensuring that 
the target state of the reachability objective is not reached, and
the game restarts. Hence for all strategies of player~$1$, there is
a strategy of player~$2$ to always avoid the target state (with probability~$1$),
and thus player~$1$ cannot win almost-surely (he wins with probability~$0$). This completes
the proof of the reduction for almost-sure winning.

Thus we obtain the following result.

\begin{theorem}\label{theo:stochastic-games}
The qualitative analysis problems (almost-sure and positive
winning) for two-player partial-observation stochastic parity games where 
player~1 is less informed than player~2, under finite-memory strategies for 
both players, are 2-EXPTIME-complete. 
\end{theorem}

\begin{remark}
Note that the lower bounds for Theorem~\ref{theo:stochastic-games} are established
for reachability objectives.
Moreover, it was shown in~\cite[Section~5]{CD12} that for qualitative analysis of 
two-player partial-observation stochastic games with reachability objectives, 
finite-memory strategies suffice, i.e., if there is a strategy to ensure 
almost-sure (resp. positive) winning, then there is a finite-memory strategy.
Thus the results of Theorem~\ref{theo:stochastic-games} hold for reachability 
objectives even without the restriction of finite-memory strategies.
\end{remark}

\end{document}